%% file: Full_version.tex
\documentclass{article}

\usepackage{versions}
\includeversion{unanonymous}

\usepackage[margin=1in,letterpaper]{geometry}
\input{commands}

\title{Quantum Key Distribution in the Classical Authenticated Key Exchange Framework}
\author{Michele Mosca$^{1,2}$, Douglas Stebila$^{3}$, and Berkant Ustaoglu$^{4}$ \\
~ \\
\small $^{1}$ Institute for Quantum Computing and Dept. of Combinatorics \& Optimization \\
\small University of Waterloo, Waterloo, Ontario, Canada \\
\small $^{2}$ Perimeter Institute for Theoretical Physics, Waterloo, Ontario, Canada \\
\small \href{mailto:mmosca@uwaterloo.ca}{\tt mmosca@uwaterloo.ca} \\
\small $^{3}$ Information Security Institute, Queensland University of Technology, Brisbane, Queensland, Australia \\
\small \href{mailto:stebila@qut.edu.au}{\tt stebila@qut.edu.au} \\
\small $^{4}$ Department of Mathematics, Izmir Institute of Technology, Urla, Izmir, Turkey \\
\small \href{mailto:bustaoglu@uwaterloo.ca}{\tt bustaoglu@uwaterloo.ca}
}
\date{June 27, 2012}

\usepackage{color}
\definecolor{darkblue}{rgb}{0,0,0.5}
\definecolor{darkgreen}{rgb}{0,0.5,0}
\usepackage[colorlinks=true,linkcolor=darkblue,urlcolor=darkblue,citecolor=darkgreen]{hyperref}

\usepackage{amsthm}
\newtheorem{definition}{Definition}
\newtheorem{theorem}{Theorem}
\newtheorem{remark}{Remark}

\begin{document} 

\maketitle
\begin{abstract}
\input{Abstract}
\end{abstract}

\input{Body}

\bibliographystyle{halphads}
\bibliography{refs}

\appendix

\input{Appendix}

\end{document}

%% file: commands.tex
\usepackage{amsmath}
\usepackage{multirow}
\usepackage{graphics}
\usepackage{subfigure}

\newcommand{\txtquery}[1]{\ensuremath{\mathsf{#1}}}
\newcommand{\Partner}{\txtquery{Partner}}

\newcommand{\SendCQ}{\txtquery{C2Q}}
\newcommand{\SendQC}{\txtquery{Q2C}}
\newcommand{\SendC}{\txtquery{SendC}}
\newcommand{\SendQ}{\txtquery{SendQ}}
\newcommand{\RevealNext}{\txtquery{RevealNext}}

\newcommand{\Test}{\txtquery{Test}}

\newcommand{\sid}[1]{\ensuremath{{\mathrm \Psi}^{#1}}}%

\newcommand{\Sign}{\ensuremath{\mathsf{Sign}}}

\newcommand{\txtvariable}[1]{\ensuremath{#1}}
\newcommand{\msg}{\txtvariable{msg}}
\newcommand{\params}{\txtvariable{params}}
\newcommand{\pid}{\txtvariable{pid}}
\newcommand{\sk}{\txtvariable{sk}}
\newcommand{\vv}{\mathbf{v}}
\newcommand{\uu}{\mathbf{u}}
\newcommand{\tapeQuantum}{q}
\newcommand{\tapeClassical}{c}
\newcommand{\tapeRandom}{r}
\newcommand{\tapeInter}{e}

\newcommand{\Succ}{\ensuremath{\mathsf{Succ}}}
\newcommand{\abort}{\ensuremath{\mathsf{abort}}}
\newcommand{\sig}{\ensuremath{\mathsf{sig}}}
\newcommand{\negl}{\ensuremath{\mathrm{negl}}}

\newcommand{\getsr}{\ensuremath{\stackrel{\$}{\gets}}}
\newcommand{\nullvalue}{\ensuremath{\perp}}
\newcommand{\numParties}{\ensuremath{n_{P}}}

\newcommand{\ket}[1]{\ensuremath{|#1\rangle}}

\usepackage{amsfonts}
\newcommand{\tickYes}{\checkmark}
\newcommand{\tickNo}{\ensuremath{\times}}
\usepackage{color}

%% file: Abstract.tex
Key establishment is a crucial primitive for building secure channels: in a multi-party setting, it allows two parties using only public authenticated communication to establish a secret session key which can be used to encrypt messages.  But if the session key is compromised, the confidentiality of encrypted messages is typically compromised as well.  Without quantum mechanics, key establishment can only be done under the assumption that some computational problem is hard. Since digital communication can be easily eavesdropped and recorded, it is important to consider the secrecy of information anticipating future algorithmic and computational discoveries which could break the secrecy of past keys, violating the secrecy of the confidential channel.

Quantum key distribution (QKD) can be used generate secret keys that are secure against any future algorithmic or computational improvements.  QKD protocols still require authentication of classical communication, however, which is most easily achieved using computationally secure digital signature schemes.  It is generally considered folklore that QKD when used with computationally secure authentication is still secure against an unbounded adversary, provided the adversary did not break the authentication during the run of the protocol.  

We describe a security model for quantum key distribution based on traditional classical authenticated key exchange (AKE) security models.  Using our model, we characterize the long-term security of the BB84 QKD protocol with computationally secure authentication against an eventually unbounded adversary.  By basing our model on traditional AKE models, we can more readily compare the relative merits of various forms of QKD and existing classical AKE protocols.  This comparison illustrates in which types of adversarial environments different quantum and classical key agreement protocols can be secure.

~


%% file: Body.tex
\section{Introduction}\label{sec:intro}

Quantum key distribution (QKD) promises new security properties compared to cryptography based on computational assumptions: QKD can provide for two parties to establish a secure key using an untrusted quantum channel and a public, authenticated classical channel, and this key is secure against any adversary who is limited solely by the laws of quantum mechanics.  While some classical\footnote{We use the adjective ``classical'' to mean ``non-quantum'', so ``classical cryptography'' means ``non-quantum cryptography'', not ``historical cryptography''.} cryptographic tasks can be achieved with information-theoretic security against unbounded adversaries, key establishment over a public authenticated channel is not one of them.  Moreover, the practicality of such information-theoretically secure schemes is often limited, and as a result most classical cryptographic schemes rely for their security on various computational assumptions, the most widely used of which --- factoring, discrete logarithms --- could be efficiently solved by a large-scale quantum computer.  As a result, QKD could be an important primitive for cryptography secure against any advances in computing technology, provided quantum mechanics remains an accurate description of the laws of nature.

\emph{Authenticated key establishment} (AKE) is the cryptographic task which QKD achieves.  The classical cryptographic literature has extensively studied AKE since the founding of public key cryptography in 1976.  After a period of ad hoc security analysis of key establishment protocols based on resistance to various individual attacks, protocols are now generally analyzed within the context of a security model, which aims to capture a wide variety of security properties in the context of an attacker who can control all communication, as well as possibly compromise participants; proofs typically consist of probabilistic reductions to computationally hard problems.  One seminal model for security of AKE protocols was proposed by Bellare and Rogaway~\cite{BR93a}.  The BR model led to the CK01 model by Canetti and Krawczyk~\cite{CK01}, upon which was based the eCK model~\cite{LLM07}.  An alternative approach to this family of security models is given by Canetti's \emph{universal composability framework}~\cite{Can01}.  One of the general observations of this line of work has been that calculating a secret key is relatively easy, but properly modelling authentication --- ensuring that the key is shared with precisely the intended party and no other --- requires greater care.

There are many types of QKD protocols, but for our purposes we will divide them into 3 classes: prepare-send-measure protocols, measure-only protocols, and prepare-send-only protocols.  The first QKD protocol, now called BB84, was proposed by Bennett and Brassard \cite{BB84}; it is an example of a prepare-send-measure protocol in which Alice randomly prepares one of several quantum states, sends it to Bob, and Bob randomly measures in one of several settings.  Ekert \cite{Eke91} proposed an entanglement-based protocol, which is an example of a measure-only protocol: Alice and Bob only randomly measure in one of several settings; the state itself can be prepared by Eve entirely untrusted.  Biham et al. \cite{BHM96} proposed a prepare-send-only protocol, in which Alice and Bob each randomly prepare one of several quantum states and send them to Eve, who measures and sends back a classical result.  Different versions can be appealing due to ease of implementation, resistance to side-channel attacks on preparing or measuring, or device independence.

Research arguing for the security of QKD has largely proceeded independent of the aforementioned classical AKE security models.  Various proofs of QKD have been given in a stand-alone 2-party setting; some of the most important ones include \cite{May96,LC99,BBBMR00,SP00,Ina02,GLLP04,Ren05}, but many others exist for different variants of QKD; some work on QKD has been done in the universal composability framework~\cite{BHLMO05}.  These proofs typically proceed under the assumption that classical communication happens over on authentic public channel; details on authenticating the classical communication are typically left out of the analysis.  It is widely recognized that the authentication can be secure against an unbounded adversary if all classical communication is protected by information-theoretically secure message authentication codes, such as the Wegman-Carter 2-universal hash function \cite{CW79,WC81}.  Alternatively, it is generally considered folklore \cite{PPS06,SECOQC07,SML09,IM11} that if QKD was performed using a computationally secure authentication scheme (such as public key digital signatures), then messages encrypted under the keys output by QKD would be secure provided that the adversary could not break the authentication scheme \emph{before or during} the QKD protocol.

\paragraph{\bf Contributions.}
Our goal is to describe the security of quantum key distribution in a security model similar to existing classical authenticated key exchange protocols and compare the relative security properties of various QKD and classical AKE protocols.  Our model is explicitly a multi-party model, includes authentication, and allows for either computationally secure or information theoretically secure authentication.  We aim to capture two properties: (1) QKD is \emph{immediately secure} against an active adversary who is restricted such that he is unable to break the authentication scheme, and (2) QKD is \emph{long-term secure}, meaning that, if it is secure against an active adversary who is restricted during the run of the protocol to be unable to break the authentication scheme, then it remains secure even when the (classical and quantum) data obtained by the active bounded adversary are subsequently given to an unbounded quantum adversary.

{\it Security model for classical-quantum AKE protocols.}
In particular, we first introduce in Section~\ref{sec:model} a multi-party model for analyzing the security of QKD protocols.  In our model, which adopts the formalism of Goldberg et al.'s framework for authenticated key exchange \cite{GSU12}, parties consist of a pair of classical and quantum Turing machines, each of which is capable of sending and receiving messages.  The adversary controls all communications between parties, but is restricted in its ability to affect communication between a single party's classical and quantum devices.  The adversary also has the ability to compromise various values used by parties during or after the run of the protocol.  As is typical, the adversary's goal is to distinguish the session key of a completed session from a random string of the same length.

Having defined the adversarial model, we then introduce our two security definitions, \emph{immediate security} against an active, potentially bounded adversary, and \emph{long-term security}, meaning security against an adversary who during the run of the protocol is potentially bounded, but after the protocol completes is unbounded (except by the laws of quantum mechanics).  Our model is generic enough to allow the bound on the adversary to be computational --- assuming that a particular computational problem is hard --- or run-time or memory-bounded~\cite{CM97}.  We adapt the long-term security notion of M\"{u}ller-Quade and Unruh \cite{MU10} from the classical universal composability framework to our classical-quantum model.  

{\it Security of BB84.}
We then proceed in Section~\ref{sec:BB84} to show that the BB84 protocol, when used with a computationally secure classical authentication scheme such as a digital signature, is secure in this model.  For the quantum aspects of the proof, we rely on existing proof techniques, but when combined with the signature scheme in our model, this work provides a proof of the folklore theorem that QKD, when used with computationally secure authentication in a multi-party setting, is information theoretically secure, provided the adversary did not break the authentication during the run of the protocol.  Note, importantly, that this is the first proof of QKD in a multi-party setting; while our QKD protocol is still a 2-party protocol, it operates in an environment where many parties may be interacting simultaneously, whereas previous proofs of security of QKD --- including the universal composability proof of Ben-Or et al. \cite{BHLMO05} --- deal with only 2 honest parties (plus the adversary).

{\it Comparison of quantum and classical AKE protocols.}
Finally, we use our generic security model to compare in Section~\ref{sec:discussion} the security properties of classical key exchange protocols and examples from each of the three classes of QKD protocols (prepare-send-measure, measure-only, prepare-send-only).  This comparison is facilitated by our phrasing of QKD in a security model more closely related to traditional AKE security models, which we can then use to compare the relative powers afforded to the adversary under those models.  In particular, our model allows us to compare how different protocols react when the randomness used in the protocol is revealed --- or if it is later discovered that bad randomness was used.  For example, some classical AKE protocol such as \texttt{UP} \cite{Ust09} is secure even if the randomness used for either a party's long-term secret key or ephemeral secret key is revealed \emph{before} the run of the protocol, but the same is not true for the randomness used to pick basis choices in BB84.  And the EPR protocol of Ekert is secure even if all of the randomness used by the parties is leaked after the protocol completes, unlike BB84 where data bit choices must remain secret.

\section{QKD model}
\label{sec:model}

Our model begins as an enhancement to the eCK model~\cite{LLM07} in which each party has access to a quantum device. The quantum device may be viewed as limited based on for example current hardware limitations. As usual we consider interactive protocols within a multi-party multi-session setting, where communication is controlled by the adversary. Subject to quantum physics restriction the adversary controls the quantum communication channel between parties.   Having described the parties and the communication model, we describe how, if at all, the adversary may gain access to secrets used by the parties.  We then define secrecy against bounded adversaries and long-term security against unbounded adversaries: the long-term security definition is achieved by having the active bounded short-term adversary output a classical and quantum transcript upon which the unbounded quantum adversary may operate.

We next formally describe the model.  We use $k$ to denote a security parameter. In the description we utilize only qubits, but these can if necessary be generalized to arbitrary-dimension quantum systems.

\subsection{Parties and protocols}
\label{sec:model:parties}

\paragraph{A party} (see also~\cite[Definition 1.1 second bullet]{ABE08}) is an interactive classical Turing machine with access to a quantum Turing machine.  Typically we refer to this pair of devices jointly as the party.

The classical machine can activate the quantum device via a special activation request or receive (via designated activation routines) measurement outcomes from the quantum device. The communication is delivered over a two way classical communication tape (the $\tapeInter$-channel in
Figure~\ref{fig:classical}). The classical Turing machine has also access to a sequence of random bits -- the $\tapeRandom$-tape in Figure~\ref{fig:classical} -- and a separate $\tapeClassical$-tape over which the party can receive and send other activation requests and messages as specified by designated routines. Similarly, the quantum Turing device can be activated by the classical Turing machine and can receive and send qubits over a designated quantum channel $\tapeQuantum$ as in Figure~\ref{fig:quantum}.  

\begin{figure}\setlength{\unitlength}{1ex}
\centering
\subfigure[Classical Turing machine]{
\begin{picture}(30,12)(0,0)
\put(4, 0){\framebox(14, 8){Alice}}
\put(11,10){\makebox(0,0){$\tapeInter$}}
\multiput(13,12)(0.5,0){2}{\vector( 0,-1){4}}
\multiput( 9, 8)(0.5,0){2}{\vector( 0, 1){4}}
\put(20,4){\makebox(0,0){$\tapeClassical$}}
\multiput(18, 2)(0,0.5){2}{\vector( 1, 0){4}}
\multiput(22, 5)(0,0.5){2}{\vector(-1, 0){4}}
\put(2,4.5){\makebox(0,0){$\tapeRandom$}}
\multiput( 0, 3)(0,0.5){2}{\vector( 1, 0){4}}
\end{picture}
\label{fig:classical}
}
\subfigure[Quantum Turing machine]{
\begin{picture}(30,12)(0,0)
\put(4, 4){\framebox(14, 8){qAlice}}
\put(11,2){\makebox(0,0){$\tapeInter$}}
\multiput(13, 4)(0.5,0){2}{\vector( 0,-1){4}}
\multiput( 9, 0)(0.5,0){2}{\vector( 0, 1){4}}
\put(20,8){\makebox(0,0){$\tapeQuantum$}}
\put(18, 6){\vector( 1, 0){4}}
\put(22, 10){\vector(-1, 0){4}}
\end{picture}
\label{fig:quantum}
}
\caption{A party's classical and quantum Turing machines}
\end{figure}
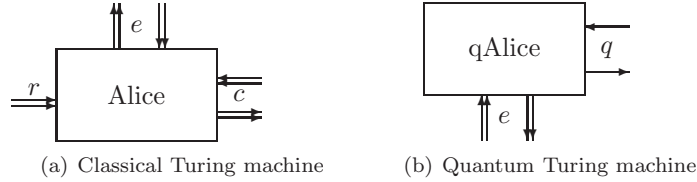

Each party can have associated authenticated public strings (which can be public keys or identifiers). Such public strings are assumed to be distributed over an authenticated channels to other parties. Furthermore, pairs of parties may possess shared secrets that were a priori distributed over a confidential and authenticated channel.

\paragraph{A protocol} is a collection of interactive classical and quantum subroutines that produce a shared secret key between two (or more parties) or output an indicator of an error. The interactions may use messages received on either the classical or quantum channels.  The final output of the protocol is made via the classical Turing machine.

\paragraph{A session} is an execution of the protocol. Sessions are initiated via a special incoming request and upon initiation each one is identified with a unique\footnote{With this definition uniqueness is guaranteed only within a party; globally uniqueness can be guaranteed by requiring the session identifier is the concatenation of the unique party identifier and the party's own session identifier} \emph{session identifier} $\sid{}$ chosen by the party at which the session is executed (in which case we say the party \emph{owns} the session). A session that has been initiated but is not yet completed is called \emph{active}. Since sessions are interactive procedures a party may own more than one active session at a given point of time. 

Each active session has a separate \emph{session state} that stores session-specific classical data.\footnote{While quantum protocols in general may make use of quantum memory for storing quantum states during a session, the current QKD protocols we consider in this paper, such as BB84 or EPR, do not, so we omit this from our model.}

Upon receiving and sending all protocol messages and performing the required measurements and computations specified by the protocol, the session \emph{completes} by having the classical Turing machine output either an error symbol $\nullvalue$ or a tuple $(\sk, \pid, \vv,\uu)$. The tuple consists of:
\begin{itemize}
\item $\sk$: a session key;
\item $\pid$: a party identifier;
\item $\vv$: a vector $(\vv_{0}, \vv_{1}, \dots)$ where each $\vv_{i}$ is a vector of public values or labels; (For example, $\vv_{1}$ may consist of the public values contributed by party $P_{1}$. Including $\vv$ as part of the session output binds the session with the various values used by the parties to compute the session key.)
\item $\uu$: a vector $(\uu_{0}, \uu_{1}, \dots)$ where each vector $\uu_{i}$ is a public value or label; $\uu$ is called the \emph{authentication} vector and indicates what information the session owner uses to identify its peer $\pid$.
\end{itemize}

\begin{definition}[Correctness]
A key exchange protocol $\pi$ is said to be \emph{correct} if, when all protocol messages are relayed faithfully, without changes to content or ordering, the peer parties output the same session key $k$ and the same vector $\vv$.
\end{definition}

\paragraph{Memory.} A party may hold in its memory several \emph{value pairs} of the form $(x, X)$, where $x$ is a private value and $X$ is a public value or label.  The pair may be a \emph{public key pair}, such as private key $x$ and public key $X$, or a \emph{labelled private value}, such as a private value $x$ and a unique public label $X=\ell(x)$.  The value pairs may be generated by some algorithm specified by the protocol. 

There are two classifications of value pairs: \emph{ephemeral} value pairs, which are associated with a particular session $\sid{}$, and \emph{static} value pairs, which can be used across multiple sessions.  The party may also have value pairs that have been generated but not yet used.  If necessary, different types of key pairs may be permitted, for example, if a protocol uses one type of key pair for digital signatures and another type of key pair for public-key encryption. The protocol specifies an algorithm for generating new pairs.  

%

\paragraph{Classical Turing machine communication.} As described above each classical Turing machine has two incoming-outgoing classical communication channels, denoted by $\tapeInter$ and $\tapeClassical$ in Figure~\ref{fig:classical}, over which the classical Turing machine receives activations and submits responses. The responses themselves can be activation requests. Furthermore the classical Turing machine has an input of classical (pseudo-)random bits which can be read at will by the Turing machine, denoted by $\tapeRandom$ in Figure~\ref{fig:classical}. 

The following activations of the classical Turing machine are allowed:

\begin{itemize}

\item $\SendC(\params, \pid)$: This activation is received via channel $\tapeClassical$ and directs the party to begin a new key exchange session. A new session is initiated and assigned a unique session identifier $\sid{}$ based on protocol-specific public parameters $\params$ and an identifier $\pid$ of the party with whom to establish the session. The response to this query includes the session identifier $\sid{}$ and any protocol-specific outgoing classical message $\msg'$ that are sent via the outgoing channel $\tapeClassical$. If required by the protocol specification the Turing machine can send an activation request $\SendCQ(m)$ over the $\tapeInter$ outgoing channel; the activation of the quantum Turing machine may cause that quantum Turing machine to write an output to its $\tapeQuantum$ channel as well, or to prepare its measurement device to receive quantum messages.

\item $\SendC(\sid{}, \msg)$: This query models the delivery of classical messages over $\tapeClassical$-channel. The party's classical Turing machine is activated with session $\sid{}$ and classical message $\msg$. It returns any outgoing classical message $\msg'$ over the $\tapeClassical$-channel. If required by the protocol specification the Turing machine can send an activation request $\SendCQ(m)$ over the $\tapeInter$ outgoing channel; the activation of the quantum Turing machine may cause that quantum Turing machine to write an output to its $\tapeQuantum$ channel as well, or to prepare its measurement device to receive quantum messages.

%
\item $\SendQC(m)$: Upon activation with this query the classical Turing machine activates its most recent session with input $m$. This query may cause the classical Turing machine to output to its $\tapeClassical$ channel, or send another activation over the $\tapeInter$ channel.
\end{itemize}

A protocol specification may request that the classical Turing machine act probabilistically. In this case the classical machine obtains random bits from the $\tapeRandom$-channel.

\paragraph{Quantum Turing machine communication.} Each party's quantum Turing machine has an incoming-outgoing classical communication channel, denoted by $\tapeQuantum$ in Figure~\ref{fig:quantum}, over which the machine receives and submits quantum information. The responses themselves can be activation requests. Furthermore the quantum Turing machine has a two-way classical control channel (denoted by $\tapeInter$ in Figure~\ref{fig:quantum}) with which it communicates with the classical Turing machine. 

The following activations of the quantum Turing machine are allowed:

\begin{itemize}

\item $\SendQ(\rho)$: This query activates the quantum Turing machine with quantum message $\rho$; it returns any outgoing quantum message $\rho'$ over the $\tapeQuantum$-channel. If required by the protocol specification the Turing machine can send an activation request $\SendCQ(m)$ over the $\tapeInter$ outgoing channel, for example, to report any measurement results obtained from measuring $\rho$.  The activation of the classical Turing machine may cause that classical Turing machine to write an output to its $\tapeClassical$ channel as well.

\item $\SendCQ(m)$: This query activates the quantum Turing machine with classical control message $m$, for example to prepare the quantum circuit for execution due to an anticipated $\SendQ$ activation. The activation may cause a quantum state to be output over the outgoing quantum channel $\tapeQuantum$ as well as a classical message to returned over the classical control channel $\tapeInter$.

\end{itemize}

\subsection{Adversarial model}
\label{sec:model:adversary}

%

\paragraph{The adversary} is, similar to a party, a pair of interactive classical and quantum Turing machines.  The adversary's classical Turing machine runs in time at most $t_{c}(k)$ and has access to a quantum Turing machine with runtime bounded by $t_{q}(k)$ and memory bounded by $m_{q}(k)$ qubits; bounds may be unlimited. The adversary takes as its input all public information and may interact with the (honest) parties.  Furthermore the adversary can establish corrupted (dishonest) parties that are fully in control of the adversary.  Honest parties are unable to distinguish between honest and dishonest parties.

Communication over the parties' classical $\tapeClassical$-channels is controlled by the adversary. On the classical channels, the adversary can
read, copy, reorder, insert, delay, modify, drop or forward messages at will. The sending and receiving parties have no intrinsic mechanism to detect which actions, if any, the adversary performed on the classical messages. 

Communication over the parties' quantum $\tapeQuantum$ channels is also controlled by the adversary.  The adversary's operations on the quantum channels are bound by the laws of quantum mechanics: the delivery of quantum messages can be delayed, modified in order, forwarded, or dropped; the adversary can create new quantum states and perform joint quantum operations on quantum messages received from the parties as well as on the adversary's state.  However, due to the laws of quantum mechanics, the adversary cannot necessarily obtain full information about quantum messages from the parties; for example, measurements by the adversary may irrevocably disturb the state of messages transmitted by the parties, and the adversary may be unable to precisely copy a message due to the no-cloning theorem. We assume that the communication channel between the adversary's quantum machine and party's quantum machines are perfect; the adversary's quantum device can simulate any environmental or noise affect on the qubits sent by a party.

\paragraph{Queries.} The adversary can direct a party to perform certain actions by sending any of the aforementioned activation queries over party's the $\tapeClassical$ and $\tapeQuantum$ channels. The adversary has neither immediate control and cannot observe the content exchanged between the classical and quantum subcomponents of a party over the $\tapeInter$ channel, nor has information about the bits obtained from the $\tapeRandom$-channel. Furthermore, to allow for information leakage the adversary may issue the following queries to parties:

\begin{itemize}

\item $\RevealNext \to X$: This query allows the adversary to activate the classical Turing machine to read input from the $\tapeRandom$-channel and learn future public values. The activated party generates a new value pair $(x, X)$, records it as unused, and returns the public value $X$.  (This query may be specialized in the event that there are multiple value pair types specified by the protocol.)

\item $\Partner(X) \to x$: This query allows the adversary to compromise secret values used in the protocol computation. If the party has a value pair
$(x, X)$ in its memory, it returns the private value $x$. $\Partner(\sid{})$ returns the secret key $\sk$ for session $\sid{}$, if it exists; this is often referred to as a \textsf{RevealSessionKey} query.

\end{itemize}

Where necessary to avoid ambiguity, we use a superscript to indicate the party to whom the query is directed, for example $\SendC^{P_{i}}(\sid{}, \msg)$.

\paragraph{Partnering.} If $(x, X)$ is a value pair, with public key value or public label $X$, then the adversary is said to be a \emph{partner} to $X$ if
the adversary issued the query $\Partner(X)$ to a party holding that value pair in its memory. Whenever a party generates a key pair $(x,X)$, for example in response to a session activation or a $\RevealNext$ query, the adversary is \emph{not} a partner to $X$ until the query $\Partner(X)$ is issued. The adversary can become a partner to any value $X$.

\subsection{Security definition}
\label{sec:model:security}

For the purpose of defining session key security, the adversary has access to the following additional oracle:

\begin{itemize}
\item $\Test(i, \sid{}) \to \kappa$: If party $P_{i}$ has not output a session key, return $\nullvalue$.  Otherwise, choose $b \getsr \{0, 1\}$.  If $b=1$, then return the session key $\sk$ from the output for session $\sid{}$ at party $P_{i}$.  If $b=0$, return a random bit string of length equal to the length of the session key $\sk$ in session $\sid{}$ at party $P_{i}$. Only one call to the $\Test$ query is allowed.
\end{itemize}

\begin{definition}[Fresh session]\label{defn:fresh}
A session $\sid{}$ owned by an honest party $P_{i}$ is \emph{fresh} if all of the following occur: 
\begin{enumerate}
\item For every vector $\vv_{j}$, $j \ge 1$, in $P_{i}$'s output for session $\sid{}$, there is at least one element $X$ in $\vv_{j}$ such that the adversary is not a partner to $X$.
\item The adversary did not issue $\Partner({\sid{}}')$ to any honest party $P_{j}$ for which ${\sid{}}'$ has the same public output vector as $\sid{}$
(including the case where ${\sid{}}' = \sid{}$ and $P_{j} = P_{i}$).
\item \emph{At the time of session completion}, for every vector $\uu_{j}$, $j \ge 1$, in $P_{i}$'s output for session $\sid{}$, there was at least one
element $X$ in $\uu_{j}$ such that the adversary was not a partner to $X$.
\end{enumerate}
\end{definition}

We emphasize the difference between the first and the third condition in the last definition: the latter is decided at the time when the session completes, whereas the former is decided at the end of the adversary's execution.  In other words, there may be some values that are okay for the adversary to learn after completion (but not before), and other values that the adversary can never learn.

\begin{definition}[Security]\label{defn:secure}
Let $k$ be a security parameter.  An authenticated key exchange protocol is \emph{secure} if, for all adversaries $\mathcal{A}$ with classical runtime bounded by $t_{c}(k)$, quantum runtime bounded by $t_{q}(k)$, and quantum memory bounded by $m_{q}(k)$, the advantage of $\mathcal{A}$ in guessing the bit $b$ used in the $\Test$ query of a fresh session is negligible in the security parameter; in other words, the probability that $\mathcal{A}$ can distinguish the session key of a fresh session from a random string of the same length is negligible.
\end{definition}

\subsection{Long-term security}\label{sec:model:long-term}

One of the main benefits of quantum key distribution is that it can be secure against unbounded adversaries. However, such strong security comes at the cost of being unable to use computationally secure cryptographic primitives such as public key digital signatures for authentication.  The definition above can be used to analyze QKD when computationally secure cryptographic primitives are used; for example, we can choose a $t_{c}(k)$, $t_{q}(k)$, and $m_{q}(k)$ such that the cryptographic primitive is believed secure against such an adversary.  The particular values may be chosen based on known classical algorithms for factoring or computing discrete logarithms and on the present-day limits of quantum devices.

Regardless of the bound on the active adversary, we can still recover a very strong form of long-term security by considering an unbounded quantum Turing machine acting after the protocol has completed.  In other words, during the run of the protocol, we assume a bounded adversary as in Definition~\ref{defn:secure}; this bounded active adversary produces some classical and quantum transcript which it then provides to the unbounded adversary.  This models the real-world scenario of an adversary being somewhat limited by its classical and quantum computing equipment now but later having much more powerful equipment or making an algorithmic breakthrough.

\begin{definition}[Long-term security]\label{defn:long-term-secure}
An authenticated key exchange protocol is \emph{long-term secure} if, for all unbounded quantum Turing machines $\mathcal{M}$ acting on a classical and quantum transcript produced by a (bounded) adversary $\mathcal{A}$ in Definition~\ref{defn:secure}, the advantage of $\mathcal{M}$ in guessing the bit $b$ used in the $\Test$ query of a fresh session is negligible in the security parameter.
\end{definition}

\subsection{Discussion}
\label{sec:model:discussion}

Several aspects of our model allow for a great range of flexibility in terms of adversarial power and allows quantum key distribution be fairly compared with classical key establishment. We will describe a few specializations of our definition and comment on one of the key differences between our model and traditional classical AKE models, the output vectors $\vv$ and $\uu$.

\paragraph{Bounds on devices.}
First, if $t_{q}(k) = m_{q}(k) = 0$, and Definition~\ref{defn:long-term-secure} is omitted, the model reduces to a classical definition for secure session key establishment. It refines the idea of authentication as the session output can explicitly identify how peers were identified and authenticated. Thus any classical protocol analyzed in~\cite{GSU12} can also be analyzed in the model presented here. The definition here is stronger in the sense it encompasses a wider range of protocols and relates to the definitions in~\cite{LLM07,CK01} the same way~\cite{GSU12} relates to them.

Secondly, it is feasible to model present limitations of quantum devices. While there are ongoing improvements in controlling quantum systems, at present the number of qubits a device can work with is essentially a small constant compared to classical computers. Thus, using our model with appropriate values of $t_{q}(k)$ and $m_{q}(k)$, based on beliefs about current practical limitations, one can devise efficient protocols that are easy to implement but guarantee unconditional future secrecy.  An appropriate assumption on $t_{c}(k)$ --- for example that all adversaries with polynomial running time $t_{c}(k)$ cannot solve a particular hard problem --- allow the model to be used as existing classical reductionist security models are used.

Of course, the devices available to the adversary can be made unbounded essentially allowing a complete quantum world. Thus the definitions presented here are suitable for analyzing novel quantum key distribution protocols.  These alternatives show the wide range of scenarios our definitions incorporate. Due to the unified underlying framework it is easier to compare various protocols and decide which one is the best for the task at hand.

\paragraph{The output vectors.}
One of the key differences between our model and traditional AKE security models is how we phrase restrictions on what secret values the adversary can learn and when.  In the eCK model, for example, a fresh session is defined as one in which the adversary has not learned (a) both the session owner's ephemeral secret key $x$ and long-term secret key $a$, and (b) both the peer's ephemeral secret key $y$ and long-term secret key $b$ (or just the peer's long-term key if no matching peer session exists).  In our model, this could be specified as $\vv = (\vv_{0}=(a, x), \vv_{1}=(b,y))$.

Since in traditional AKE security models the restriction on values learned is specified in the security model, a new security model is required for each differing combination of learnable values.  Though models may often appear similar, they sometimes contain subtle but important formal differences and thus become formally incomparable \cite{Cre11}.  The traditional approach of specifying the values that can or cannot be learned in the security definition itself contrasts with our approach --- building on that of Goldberg~et al.~\cite{GSU12} --- where the vectors $\vv$ and $\uu$ in the session output specify what can or cannot be learned.  As a result, two protocols with differing restrictions on values that can be learned could both be proven secure in our model and then compared based on which values can or cannot be revealed.  

\section{BB84}
\label{sec:BB84}

We now turn to BB84 protocol~\cite{BB84}. We first specify the protocol in the language of the model of Section~\ref{sec:model}, discuss some aspects of our formulation, and complete the section with a security analysis.

\begin{definition}
Let $k$ be a security parameter.  The \emph{BB84 protocol} is defined by having parties responding to activations as follows:
\begin{enumerate}
\item Upon activation $\SendC(\mbox{\tt start}, \mbox{\tt initiator}, B)$ the classical Turing
machine $A$ does the following:
	\begin{enumerate}
	\item create a new session $\sid{A}$ with peer identifier $B$;
	\item read $n_{1}$ (random) data bits $\sid{A}_{dAB}$ from its $\tapeRandom$-tape;
	\item read $n_{1}$ (random) basis bits $\sid{A}_{bA}$ from its $\tapeRandom$-tape;
	\item send activation $\SendCQ(\sid{A}_{bA}, \sid{A}_{dAB})$ on its $\tapeInter$-tape, which indicates that the quantum device should encode each data bit from $\sid{A}_{dAB}$ as $\ket{0}$ or $\ket{1}$ if the corresponding basis bit $\sid{A}_{bA}$ is 0, or as $\ket{+}$ or $\ket{-}$ if the corresponding basis bit $\sid{A}_{bA}$ is 1;
	\item send activation $\SendC(\sid{A},\mbox{\tt start}, \mbox{\tt responder}, A)$ on its $\tapeClassical$-tape to $B$.
	\end{enumerate}

\item Upon activation $\SendC(\sid{A},\mbox{\tt start}, \mbox{\tt responder}, A)$ the classical Turing machine $B$ does the following:
	\begin{enumerate}
	\item create a new session $\sid{B}$ with peer identifier $A$;
	\item read $n_{1}$ (random) basis bits $\sid{B}_{bB}$ from its $\tapeRandom$-tape;
	\item send activation $\SendCQ(\sid{B}_{bB})$ on its $\tapeInter$-tape, which indicates the quantum device should measure the $i$th qubit in the $\ket{0} / \ket{1}$ if the $i$th bit of $\sid{B}_{bB}$ is 0, or in the $\ket{+} / \ket{-}$ basis if $i$th bit of $\sid{B}_{bB}$ is 1.
	\end{enumerate}

\item Upon activation $\SendQC(m)$, the classical Turing machine $B$ does the following:
	\begin{enumerate}
	\item set $\sid{B}_{dAB}$ equal to $m$;
	\item compute $\sigma \gets \Sign_{pk_{B}}(\sid{A},\sid{B},\sid{B}_{bB},B)$;
	\item send activation $\SendC(\sid{A},\sid{B},\sid{B}_{bB},\sigma)$ on its $\tapeClassical$-tape to $A$.
	\end{enumerate}

\item Upon activation $\SendC(\sid{A},\sid{B},\sid{B}_{bB},\sigma)$, the classical Turing machine $A$ does the following:
	\begin{enumerate}
	\item verify $\sigma$ with $pk_{B}$;
	\item discard all bit positions from $\sid{A}_{dAB}$ for which $\sid{A}_{bA}$ is not equal to $\sid{B}_{bB}$; assume there are $n_{2}$ such positions left;
	\item read $n_{2}$ (random) bits $\sid{A}_{indAB}$ from its $\tapeRandom$-tape; set $\sid{A}_{chkAB}$ to be the substring of $\sid{A}_{dAB}$ for which the bits of $\sid{A}_{indAB}$ are 1, and set $\sid{A}_{kAB}$ to be the substring of $\sid{A}_{dAB}$ for which the bits of $\sid{A}_{indAB}$ are 0; let $n_{3}$ denote the length of $\sid{A}_{kAB}$
	\item compute
	$\sigma \gets \Sign_{pk_{A}}(\sid{A},\sid{B},\sid{A}_{bA},\sid{A}_{indAB},\sid{A}_{chkAB},A)$;
	\item send activation $\SendC(\sid{A},\sid{B},\sid{A}_{bA},\sid{A}_{indAB},\sid{A}_{chkAB},\sigma)$ on its $\tapeClassical$-tape to $B$.
	\end{enumerate}

\item Upon activation $\SendC(\sid{A},\sid{B},\sid{A}_{indAB},\sid{A}_{chkAB},\sigma)$, the classical Turing machine $B$ does the following:
	\begin{enumerate}
	\item verify $\sigma$ with $pk_{A}$;
	\item discard all bit positions from $\sid{B}_{dAB}$ for which $\sid{A}_{bA}$ is not equal to $\sid{B}_{bB}$
	\item set $\sid{B}_{chkAB}$ to be the substring of $\sid{B}_{dAB}$ for which the bits of $\sid{A}_{indAB}$ are 1, and set $\sid{B}_{kAB}$ to be the substring of $\sid{B}_{dAB}$ for which the bits of $\sid{A}_{indAB}$ are 0
	\item let $\epsilon$ be the proportion of bits of $\sid{A}_{chkAB}$ that do not match $\sid{B}_{chkAB}$; if $\epsilon > 0.061$ then abort;
	\item compute $\sigma \gets \Sign_{pk_{B}}(\sid{A},\sid{B},\epsilon,B)$;
	\item send activation $\SendC(\sid{A},\sid{B},\epsilon,\sigma)$ on its $\tapeClassical$-tape to $A$.
	\end{enumerate}

\item Upon activation $\SendC(\sid{A},\sid{B},\epsilon,\sigma)$, the classical Turing machine $A$ does the following:
	\begin{enumerate}
	\item verify $\sigma$ with $pk_{B}$;
	\item read (random) bits $\sid{A}_{F}$ from its $\tapeRandom$-tape to construct a random a 2-universal hash function $F : \{0,1\}^{n_{3}} \to \{0,1\}^{r'}$ (where $r' = n_{3}h(\epsilon) + o(n_{3})$) for information reconciliation (see Appendix~\ref{app:IR-PA}) and compute $F' = F(\sid{A}_{kAB})$;
	\item read (random) bits $\sid{A}_{P,G}$ from its $\tapeRandom$-tape to generate a random permutation $P$ on $n_{3}$ elements and a 2-universal hash function $G : \{0,1\}^{n_{3}} \to \{0,1\}^{s'}$ (where $s' = n_{3}(1-3h(\epsilon)) + o(n_{3})$) for privacy amplification (see Appendix~\ref{app:IR-PA}), respectively; compute $\sid{A}_{skAB} \gets G(P(\sid{A}_{kAB}))$;
	\item compute $\sigma \gets \Sign_{pk_{A}}(\sid{A},\sid{B},F,F',P,G,A)$;
	\item send activation $\SendC(\sid{A},\sid{B},F,F',P,G,\sigma)$ on its $\tapeClassical$-tape to $B$;
	\item output $(\sk = \sid{A}_{skAB}, \pid = B, \vv = (\vv_{0} = (\ell(\sid{A}_{dAB})),\vv_{1} = (\ell(\sid{A}_{bAB})), \vv_{2} = (\ell(\sid{B}_{dAB})),\vv_{3}=(\ell(\sid{B}_{bAB})),\vv_{4}=(\ell(\sid{A}_{F})),\vv_{5}=(\ell(\sid{A}_{P,G}))), \uu = (\uu_{1} = (pk_{B})))$ (recall $\ell(\cdot)$ denotes the label describing the corresponding secret value).
	\end{enumerate}

\item Upon activation $\SendC(\sid{A},\sid{B},F,F',P,G,\sigma)$, the classical Turing machine $B$ does the following:
	\begin{enumerate}
	\item verify $\sigma$ with $pk_{A}$;
	\item use $F$ and $F'$ to correct $\sid{B}_{kAB}$ to $\sid{B}_{kAB'}$;
	\item compute $\sid{B}_{skAB} \gets G(P(\sid{B}_{kAB'}))$;
	\item output $(\sk = \sid{B}_{skAB}, \pid = A, \vv = (\vv_{0} = (\ell(\sid{A}_{dAB})),\vv_{1} = (\ell(\sid{A}_{bAB})), \vv_{2} = (\ell(\sid{B}_{dAB})), \vv_{3}=(\ell(\sid{B}_{bAB})),\vv_{4}=(\ell(\sid{A}_{F})),\vv_{5}=(\ell(\sid{A}_{P,G})),), \uu = (\uu_{1} = (pk_{A})))$.
	\end{enumerate}
\end{enumerate}
\end{definition}

\begin{remark}
In the output vector $\vv$, the values $\ell(\sid{A}_{bAB})$, $\ell(\sid{B}_{bAB})$, $\ell(\sid{A}_{F})$, and $\ell(\sid{A}_{P,G})$ appear as single component vectors. But in step 6(e) the values are broadcast in the clear. This may seem a bit contradictory since, if the adversary becomes a partner to either of those values (and therefore learns their content), the session is not fresh, but because of the broadcast the adversary \emph{does} in fact learn the values corresponding to the aforementioned labels. The important distinction is \emph{when} the adversary obtains these values, either before or after the protocol commences and measurements are performed. For the adversary to learn these values before parties' measurements, it must partner to these values, violating session freshness. Learning the values after the session completes is not an issue and the values are given to the adversary ``for free'', without the need for partnering.
\end{remark}

\begin{remark}
The output vector $\uu$ represents the values which the session owner uses to authenticate its peer. Similar to $\ell(\sid{A}_{bAB})$ the authentication information has to be exclusively available to the alleged peer, but only at the time of protocol execution: these values may subsequently be revealed. Therefore, as described in Definition~\ref{defn:fresh}, partnering to $\uu$ is decided upon session completion. 

Observe that for the BB84 protocol above, Alice does not include her own authentication secret $pk_{A}$. This implies that the protocol is resilient to \emph{key compromise impersonation attacks}: even with Alice's authentication keys no party is able to pretend to be someone other than Alice to Alice.
\end{remark}

\subsection{Security of BB84}\label{sec:BB84:proof}

We now show that the BB84 protocol stated above is a secure and long-term-secure authenticated key exchange protocol under the assumption that the bounded active adversary cannot break the signature scheme.

\begin{theorem}[Security of BB84]
\label{thm:BB84-short-term}
Let $k$ be a security parameter.  Suppose that the probability $\epsilon_{\sig}$ that any probabilistic polynomial time classical Turing machine with oracle access to a $(t_{q}(k), m_{q}(k))$-bounded quantum Turing machine can break the signature scheme is negligible in $k$.  Then the BB84 protocol is a secure authenticated key exchange protocol according to Definition~\ref{defn:secure}.
\end{theorem}

\emph{Proof sketch.} Our proof combines an existing proof of security by Christandl et al. \cite{CRE04} for the BB84 protocol with the sequence-of-games technique of Shoup \cite{Sho06}.  First we show --- using techniques from classical reductionist security --- that no bounded adversary can (except with negligible probability) successfully tamper with the classical authenticated communication.  Then we show --- using techniques from QKD security proofs --- that the adversary cannot distinguish the key from random.  Details appear in Appendix~\ref{app:proof}.

\begin{theorem}[Long-term security of BB84]
\label{thm:BB84-long-term}
Let $k$ be a security parameter.  Suppose the signature scheme is secure against all bounded adversaries as specified in Theorem~\ref{thm:BB84-short-term}.  Then the BB84 protocol is a long-term secure authenticated key exchange protocol according to Definition~\ref{defn:long-term-secure}.
\end{theorem}

\begin{proof}
The argument has in fact appeared in the argument of Theorem~\ref{thm:BB84-short-term}. Observe that in its proof the bounds on $t_{c}(k)$, $t_{q}(k)$, and $m_{q}(k)$ and on the adversary general is required only for guaranteeing the authenticity and origin of messages in the hop from game 0 to game 1. Once assured that the classical authentic communication has not been tampered with the remainder of the argument is a typical argument for a quantum key distribution scheme, which does not require any bounds on the adversarial power. Since the unbounded adversary runs after the protocol completes, meaning it cannot inject reorder or modify messages in the transcript, the past classical communication remains authentic and the result follows.
\end{proof}

\section{Comparing classical and quantum key exchange protocols}\label{sec:discussion}

In Section~\ref{sec:model:discussion}, we discussed how our model can be used to analyze both purely classical protocols and quantum protocols. Given its similarity to existing classical AKE security models and its flexibility in analyzing the security of a variety of protocols, it is natural to use the model to try to identify qualitative differences between different classes of protocols.

One of the key differences between existing AKE security models such as CK01 and eCK is what randomness the adversary is allowed reveal, and when, and still have the protocol be secure.  Our framework is more generic: it is not the \emph{model} that specifies which randomness can be revealed but the \emph{protocol itself} in its output vectors $\vv$ and $\uu$.  As a result, we can ``compare'' protocols by viewing them all within our model and then comparing which values are included in the output vector.\footnote{We note that it has been shown \cite{Cre11} that the CK01 and eCK models are \emph{formally incomparable}, meaning neither can be shown to imply the other.  Nonetheless, properties of \emph{specific protocols} secure in those models may be compared by resorting to a third model such as the one in this paper.}

Table~\ref{fig:comparison} summarizes the observations of this section.  We compare are two qualitatively different classical AKE protocols and three qualitatively different QKD protocols: (1) the signed Diffie--Hellman protocol \cite{CK01} (which can be proven secure in the CK01 model), (2) the \texttt{UP} protocol \cite{Ust09}, a variant of the MQV protocol \cite{LMQSV03} which can be proven secure in the eCK model, (3) the BB84 \cite{BB84} prepare-send-measure QKD protocol, (4) the EPR \cite{Eke91} (entanglement-based) measure-only QKD protocol, and (5) the BHM96 \cite{BHM96,Ina02} prepare-send-only QKD protocol.  We note that our model is flexible enough to allow all these protocols to be proven secure in it, of course with different cryptographic assumptions, bounds on the adversary, and different output vectors.  It is these differences we compare in Table~\ref{fig:comparison}

\begin{table}[tb]
\caption{Comparison of security properties of various classical and quantum AKE protocols.}
\label{fig:comparison}
\centering
\scalebox{0.75}{
\begin{tabular}{l || c | c | c | c | c}
\multirow{2}{*}{Protocol} 	& Signed Diffie-- 		& \texttt{UP} 			& BB84 				& EPR 				& BHM96 \\ 
					& Hellman \cite{CK01} 	& \cite{Ust09} 			& \cite{BB84} 			& \cite{Eke91} 			& \cite{BHM96,Ina02} \\ \hline \hline
Protocol type 			& \multirow{2}{*}{classical} & \multirow{2}{*}{classical} & quantum 			& quantum 			& quantum \\
					& 					& 					& prepare-send-measure	& measure-only			& prepare-send-only \\ \hline
Security model in which	& CK01 \cite{CK01}, 		& eCK \cite{LLM07},	 	& \multirow{2}{*}{this paper} & \multirow{2}{*}{this paper} & \multirow{2}{*}{this paper} \\
can be proven secure	& this paper			& this paper			&					&					& \\ \hline
Randomness revealable 	& $\tickNo$ static key 	& at most 1 of	 		& $\tickNo$ static key 	& $\tickNo$ static key 	& $\tickNo$ static key \\
{\bf before} protocol run?	& $\tickNo$ ephemeral key & static key, 			& $\tickNo$ basic choice 	& $\tickNo$ basis choice 	& $\tickNo$ basis choice \\	
					&					& ephemeral key	 	& $\tickNo$ data bits		& 					& $\tickNo$ data bits \\
					&					&					& $\tickNo$ info. recon.	& $\tickNo$ info. recon. 	& $\tickNo$ info. recon. \\
					&					&					& $\tickNo$ priv. amp.	& $\tickNo$ priv. amp. 	& $\tickNo$ priv. amp. \\ \hline
Randomness revealable 	& $\tickYes$ static key 	& at most 1 of	 		& $\tickYes$ static key 	& $\tickYes$ static key 	& $\tickYes$ static key \\
{\bf after} protocol run?	& $\tickNo$ ephemeral key & static key, 			& $\tickYes$ basis choice 	& $\tickYes$ basis choice 	& $\tickYes$ basis choice \\
					&					& ephemeral key 		& $\tickNo$ data bits		& 					& $\tickNo$ data bits \\
					&					&					& $\tickYes$ info. recon.	& $\tickYes$ info. recon. 	& $\tickYes$ info. recon. \\
					&					&					& $\tickYes$ priv. amp.	& $\tickYes$ priv. amp. 	& $\tickYes$ priv. amp. \\ \hline
Short-term security		& computational		& computational		& computational or 		& computational or		& computational or \\
					& assumption			& assumption			& information-theoretic	& information-theoretic	& information-theoretic \\ \hline
Long-term security		& $\tickNo$			& $\tickNo$			& assuming short-term-	& assuming short-term-	& assuming short-term- \\
					&					&					& secure authentication	& secure authentication	& secure authentication 
\end{tabular}
}
\end{table}

{\it Revealing randomness before the run of the protocol.}
Some classical AKE protocols, especially eCK-secure protocols such as \texttt{UP} and similar MQV-style protocols, remain secure even if the adversary learns either the ephemeral secret key or the long-term secret key, but not both, before the run of the protocol.  This contrasts with all known QKD protocols, where none of the random values used during the protocol --- the long-term secret key, the basis choices (for measure protocols), data bits (for prepare protocols), information reconciliation function, or privacy amplification function --- can be revealed to the adversary in advance.  (This is why all of these values are included individually in the output vector $\vv$ in the BB84 specification in Section~\ref{sec:BB84}.)

{\it Revealing randomness after the run of the protocol.}
For classical AKE protocols to remain secure, at least some secret values must not be revealed after the run of the protocol.  For protocols with so-called perfect forward secrecy, such as signed Diffie--Hellman, the parties' long-term secret keys can be corrupted after the run of the protocol, but not the ephemeral secret keys.  For eCK-secure protocols such as \texttt{UP} and similar MQV-style protocols, either the long-term secret key or the ephemeral secret key, but not both, can be revealed after the protocol run (or, as per the previous paragraph, before/during).  For measure-only entanglement-based QKD protocols such as EPR, all random choices made by the parties can be revealed after the run of the protocol: this is because the key bits are not chosen by the parties, nor in fact by the adversary, but are the result of measurements and (after successful privacy amplification) are uncorrelated with any of the input bits of any of the parties, including the adversary.  This is not the case for prepare-and-send protocols such as BB84 or BHM96, as the sender does randomly choose data bits which must remain secret.

{\it Short-term and long-term security.}
Classical AKE protocols can be proven secure only under computational assumptions, and as such only offer short-term security in the sense of Definition~\ref{defn:secure}.  Even against an unbounded passive adversary they do not retain any of their secrecy properties.  Thus classical AKE protocols are only secure against bounded short-term adversaries; however, they can be compared on the relative strength of the bound on the adversary.  This contrasts with QKD protocols.  QKD can be shown to be secure against either \emph{unbounded} short-term adversaries, by using information-theoretic authentication, or secure against bounded short-term adversaries when using a computationally secure authentication scheme as we have shown for BB84 in Section~\ref{sec:BB84:proof}.  A key contribution of the model in Section~\ref{sec:model} is a formalism which captures the notion that QKD can remain secure against an unbounded adversary after the protocol completes, provided the adversary at the time of the run of the protocol could not break the authentication scheme.

We note that applications wishing to achieve both the long-term security properties of QKD and the resistance to randomness revelation that eCK-secure classical AKE protocols have could do so by running both protocols in parallel for each session, and then combining the keys output by the two protocols together; if combined correctly, the resulting key would provide strong short-term security and strong long-term security.  This approach is indeed being used by QKD implementers, such as commercial QKD vendor ID Quantique.\footnote{\url{http://www.idquantique.com/images/stories/PDF/cerberis-encryptor/cerberis-specs.pdf}}

\section{Conclusions}\label{sec:concl}

We have presented a model for key establishment which incorporates both classical key agreement and quantum key distribution. Our model can accommodate a wide range of practical and theoretical scenarios and can serve as a common framework in which to compare relative security properties of different protocols.  A key aspect of our model is that restrictions on values that the adversary can compromise are not specified by the model but by the output of the protocol.  Using our model, we were able to provide a formal argument for the short-term and long-term security of BB84 in the multi-user setting while using computationally secure authentication.  

The ability to compare various classical and quantum protocols in our model has allowed us to identify an important distinction between existing classical key establishment and quantum key distribution protocols.  At a high level, classical protocols can provide more assurances against online adversaries who can leak or infiltrate in certain ways, but in the long run may be insecure against potential future advances. Current quantum protocols provide assurances against somewhat weaker online adversaries but retain secrecy indefinitely, even against future advances in computing technology. 

Since in our model the relative strength of a fresh session is specified by the conditions given in the output vector, an interesting open problem would be to use our model develop a quantum key distribution protocol which does retain its security attributes in the short- and long-terms even if some random values were known before the run of the protocol.  Also of interest is how to best combined keys from both quantum and classical key exchange protocols run in parallel.

\begin{unanonymous}

\subsection*{Acknowledgements}

The authors gratefully acknowledge helpful discussions with Norbert L\"{u}tkenhaus, Alfred Menezes, and Kenny Paterson.

MM is supported by NSERC (Discovery, SPG FREQUENCY, CREATE), QuantumWorks, MITACS, CIFAR, ORF. IQC and Perimeter
Institute are supported in part by the Government of Canada and the Province of Ontario.

\end{unanonymous}

%% file: Appendix.tex
\section{Information reconciliation and privacy amplification}
\label{app:IR-PA}

See \cite[\S4.4.1--4.4.3]{CRE04} for a formal analysis of information reconciliation and privacy amplification in the context of quantum key distribution.

\subsection{2-universal hash functions}
A family of \emph{2-universal hash functions} is a set of hash functions $\mathcal{H}$ mapping a set $U$ to bit strings of length $r'$ if, for all $x, y \in U$ with $x \ne y$, 
\[ \Pr_{H \in \mathcal{H}} \left( H(x) = H(y) \right) \le 2^{-r'} \enspace . \]
An example of a 2-universal hash function is as follows.  Fix $r'$.  Let $U=\{0, 1, \dots, 2^{w}-1\}$, with $w > r'$.  Let $a$ be a randomly chosen positive odd integer with $a < 2^{w}$ and let $b=i2^{w/2}$ where $i$ is chosen at random from $\{0, \dots, 2^{w/2}-1\}$.  Define 
\[ H_{a, b}(x) = ((ax+b) \mod 2^{w}) \mathop{\mathrm{div}} 2^{w-r'} \]
where $\mathrm{div}$ denotes integer division.  Then $\mathcal{H} = \{ H_{a, b} : a, b \mbox{ as above} \}$ is a family of 2-universal hash functions \cite{DHKP97}.

\subsection{Using 2-universal hash functions for information reconciliation}

Let $\epsilon$ be the proportion of Bob's check bits $ckb_{B}$ that disagree with Alice's check bits $ckb_{A}$. Set $r = \lceil nh(\epsilon) \rceil$.  Choose $r' = r+o(n)$.  Choose $F \in_{R} \mathcal{H}_{r'}$.  Alice sends to Bob the description of the function $F$ and the value $F(\sid{A}_{kAB})$.  Bob corrects $\sid{B}_{kAB}$ to $\sid{A}_{kAB}$ by guessing the errors and checking based on the received value $F(\sid{A}_{kAB})$.

Note that information reconciliation can also be achieved using certain types of error correcting codes.

\subsection{Using 2-universal hash functions for privacy amplification}

Alice chooses a random permutation $P$ on $|\sid{A}_{kAB}|$ elements.  She also chooses a random 2-universal hash function $G$ that maps $|\sid{A}_{kAB}|$ bits to $s'=n-\lceil 3nh(\epsilon) \rceil+o(n)$ bits, where $\epsilon$ is (as before) the proportion of Bob's check bits $ckb_{B}$ that disagree with Alice's check bits $ckb_{A}$.  Alice sends $P$ and $G$ to Bob authentically, for example by signing it.  Alice computes her final session key as $\sid{A}_{skAB}=G(P(\sid{A}_{kAB}))$ and Bob computes his final session key as $\sid{B}_{skAB}=G(P(\sid{B}_{kAB}))$.

\section{Proof of Theorem~\ref{thm:BB84-short-term}}\label{app:proof}

\begin{proof}
Our proof combines an existing proof of security by Christandl et al. \cite{CRE04} for the BB84 protocol with the sequence-of-games technique of Shoup \cite{Sho06}.\footnote{In the sequence-of-games technique, we make small changes to the security experiment, one after the other, beginning with the original security experiment.  We must show that no adversary can distinguish any of the individual changes we made, and then that the final version of the experiment that we reach is secure.}  First we show --- using techniques from classical reductionist security --- that no bounded adversary can (except with negligible probability) successfully tamper with the classical authenticated communication.  Then we show --- using techniques from QKD security proofs --- that the adversary cannot distinguish the key from random.

Let $\Succ_{i}$ denote the event that, in game $i$, the adversary successfully guesses the bit $b$ used in the $\Test$ query against a fresh session.

\paragraph{Game 0.}
This is the original security experiment.  Our goal is to prove an upper bound on $\left| \Pr \left( \Succ_{0} \right) - \frac{1}{2} \right|$.

\paragraph{From Game 0 to Game 1.}
In this game, we want to ensure that all parties that output session keys receive as input over the classical $\tapeClassical$-channel exactly the messages sent by its peer's session.  We make use of the fact that each party chooses its session identifier uniquely (within itself) and that these session identifiers are included in every digital signature.

Let $\abort_{\sig}$ be the event that there exists an honest party $P$ owning a fresh session $\sid{A}$ that output a session key such that
\begin{itemize}
\item party $P$ received $\sid{B}$ as the session identifier of the peer's session,
\item there there is no honest party $P'$ with session identifier $\sid{B}$ and peer session identifier $\sid{A}$,
\item but the signature received by party $P$ in either step 4(a) (if $P$
is ``Alice''), step 5(a) (if $P$ is ``Bob''), 6(a) (if $P$ is ``Alice''), or 7(a) (if $P$ is ``Bob'') verifies correctly under the long-term public key of the party corresponding to the peer identifier $pid$ of the session.
\end{itemize}
If $\abort_{\sig}$ occurs, the challenger aborts.

We have that $\left| \Pr \left( \Succ_{0} \right) - \Pr \left( \Succ_{1} \right) \right| \le \Pr \left( \abort_{\sig} \right)$.  

We now need a bound on $\Pr \left( \abort_{\sig} \right)$.  We will obtain such a bound by constructing a signature forger as follows.  The forger receives as input a public key $pk^{*}$ and simulates the challenger for the adversary.  The challenger guesses an index $i^{*}$ of a party, and generates all public keys / secrets for all parties except party $i^{*}$ as before.  The challenger then proceeds exactly as in game 0, except that whenever party $P_{i^{*}}$ is required to generate a signature on a message $m$, the challenger uses the signing oracle of the signature challenger.  

Suppose event $\abort_{\sig}$ occurs at some party $P_{i}$ in a session with
peer identifier $P_{j}$.  This means that party $P_{i}$ has received as input
a signature on either $(\sid{A},\sid{B},\sid{B}_{bB},B)$ (in step 4(a)),
$(\sid{A},\sid{B},\sid{A}_{bA},\sid{A}_{indAB},\sid{A}_{chkAB},\allowbreak A)$ (in
step 5(a)), $(\sid{B},\sid{A},\epsilon,B)$ (in step 6(a)), or
$(\sid{A},\sid{B},F,F',P,G,A)$ (in step 7(a)) but no session at party
$P_{j}$ ever issued the corresponding signature, since there is session with
identifier $\sid{B}$ at $P_{j}$.  

If $i^{*} = j$, which happens with non-negligible probability $1/\numParties$, where $\numParties$ is the number of parties, then the forger can use the signature $\sigma$ received by $P_{i}$ to break the existential-unforgeability-under-chosen-message-attack of the signature scheme with success probability at least $4\epsilon_{\sig}$.  Thus,
\begin{equation}\label{eqn:G0-G1}
\left| \Pr \left( \Succ_{0} \right) - \Pr \left( \Succ_{1} \right) \right| \le 4 \cdot \numParties \cdot \epsilon_{\sig} \enspace .
\end{equation}

\paragraph{Game 1.} Having assured that classical communication is untampered with in honest sessions that output session keys, we now make use of standard proofs for security of quantum key distribution.  In particular, we follow the technique of Christandl et al.~\cite{CRE04}.  We provide a brief sketch of their argument.

Let $\varepsilon \ge 0$ be negligible in $k$.  The \emph{$\varepsilon$-smooth R\'{e}nyi entropy of order $\infty$ of a probability distribution $P$} is denoted by $H_{\infty}^{\varepsilon}(P)$ \cite[Definition 3.5]{CRE04}.  For intuitive purposes, we will refer to this as simply the ``sR-entropy'' of $P$; the detailed analysis appears in \cite[\S4.4--5.1]{CRE04}.

The sR-entropy of $\sid{A}_{kAB}$ is $n_{3}$ bits.  After transmitting the information reconciliation value $F'=F(\sid{A}_{kAB})$ in step 6, $A$ reveals at most an additional $r' = n_{3}h(\epsilon)+o(n)$ bits of sR-entropy.

Due to the laws of quantum mechanics, any attacker observing or modifying will alter $A$'s quantum transmission in proportion to the amount of information the attacker gains.  In particular, we can obtain an upper bound on the amount of information learned by the attacker about $\sid{A}_{kAB}$ based on the proportion of errors in the check bits $\sid{A}_{chkAB}$ and $\sid{B}_{chkAB}$.  If $\epsilon$ is the proportion of errors, then the amount of information learned by any attacker is (except with probability exponentially small in $n_{3}$) upper-bounded by $2n_{3}h(\epsilon)$.

Thus, the attacker's sR-entropy of $\sid{A}_{kAB}$ conditioned on her attack is at least $n_{3}(1-3h(\epsilon))+o(n_{3})$.  By applying privacy amplification in step 6(c) and 7(c), Theorem~4.7 of \cite{CRE04} implies that the probability distribution on the final session key $\sid{A}_{skAB}$ is $\delta$ close to uniform, where $\delta \le 3 \cdot 2^{-\frac{n_{3}-r-s+1}{2}} + \negl(k)$, which is negligible in $k$.  Thus no attacker can distinguish $\sid{A}_{skAB}$ from a uniformly random string of the same length except with negligible probability.  This argument shows that
\begin{equation}\label{eqn:G1}
\Pr \left( \Succ_{1} \right) \le \negl(k) \enspace .
\end{equation}

Combining equations~(\ref{eqn:G0-G1}) and (\ref{eqn:G1}), we obtain our result that $\left| \Pr \left( \Succ_{0} \right) - \frac{1}{2} \right| \le \negl(k)$.
\end{proof}